\documentclass[runningheads]{llncs}
\usepackage{graphicx}
\usepackage{amsfonts}
\usepackage{amsmath}
\usepackage{hyperref}

\begin{document}
\title{Abelian-square factors and binary words}
%
%
\author{Salah Triki
}
%
\authorrunning{S. Triki}
%
\institute{
Mir@cl Laboratory\\
University of Sfax, Tunisia \\
\email{salah.triki@fsegs.rnu.tn}
}
\maketitle              
\begin{abstract}
In this work, we affirm the conjecture proposed by Gabriele Fici and Filippo Mignosi in \cite{FiciM15}.
\end{abstract}
%
%
%






\section{Definitions}
\begin{definition}
A word $w$ is called a $factor$ of a word $u$ if there exists words x, y such that $u=xwy$.
\end{definition}
\begin{definition}
An abelian-square is a word of length $2n$ where the first $n$ symbols form an anagram of the last $n$ symbols.
\end{definition}

\section{Conjecture}
\begin{lemma}
Let $w$ is a word of length $n$, containing $k$
many distinct abelian-square factors, and with the last symbol is in an abelian-square factor. Then a binary word of length $n$
containing at least $k$ many distinct abelian-square factors, and with the last symbol is in an abelian-square factor, exists.
\end{lemma}
The binary word will be called a {\itshape binary image} of $w$.
\begin{proof}
By induction on $n$. For $n<=2$, the claim is clear.\\
Assume that the claim holds for a word $w$ of length $n$ and $w'$ is a binary image of $w$.
So, $wx$ with $x$ equals to the last symbol of $w$, has $k$
many distinct abelian-square factors, and has a length $n+1$. And, $w'y$ with $y$ equals to the last symbol of $w'$
is a binary image of $wx$ \qed
\end{proof}

\begin{conjecture} \cite{FiciM15}
Assume that a word with length $n$, and containing $k$
many distinct abelian-square factors, exists. Then a binary word of length $n$
containing at least $k$ many distinct abelian-square factors exists.
\end{conjecture}
\begin{proof}
By induction on $n$. For $n<=2$, the claim is clear.\\
Assume that the claim holds for a word $w$ of length $n$. 
So, if the last symbol of $w$ is in a factor, then by lemma 1, $wx$ with $x$ equals to the last symbol of $w$ has a binary image of length $n+1$ with at least $k$ distinct abelian-square factors. If the last symbol is not in a factor, then also by lemma 
1, $wx$ has a binary image with at least $k+1$ distinct abelian-square factors

\qed
\end{proof}




%


\section*{Acknowledgements}
The author acknowledges, the financial support of this work from the Tunisian General Direction of Scientific Research (DGRST).

%
%
%
 \bibliographystyle{splncs04}
 \bibliography{bilio}
%





\end{document}